\documentclass[letterpaper, 10 pt, conference]{ieeeconf} 

\usepackage{mathtools}
\usepackage{amsmath,mleftright}
\usepackage{xparse}
\usepackage{comment}
\usepackage{amssymb}
\usepackage{xcolor}
\usepackage{mathbbol}
\usepackage{graphicx}% http://ctan.org/pkg/graphicx
\usepackage{amsmath}
\usepackage{amssymb}
\usepackage{cite}
\usepackage[colorlinks=true,urlcolor=blue]{hyperref}  
\usepackage{float}

\usepackage{subfig}

\newtheorem{theorem}{Theorem}

\newtheorem{assumption}{Assumption}

\hypersetup{
    colorlinks=true,
    linkcolor=blue,
    filecolor=magenta,      
    urlcolor=cyan,
    pdftitle={Overleaf Example},
    pdfpagemode=FullScreen,
    }

\begin{document}
%\title{Control Barrier Functions with \\Circulation Inequalities}
%\title{A Circulation Approach to Mitigate Spurious Equilibria in Control Barrier Function Based Methods}
%\title{Mitigating Spurious Equilibria in Control Barrier Function Methods Using Circulations}
\title{Using Circulation to Mitigate Spurious Equilibria in Control Barrier Function}
%\title{Avoiding Undesirable Equilibria in Control Barrier Function Approaches for Multi-Robot Planar Systems}
\author{Vinicius Mariano Gonçalves$^{1}$, Prashanth Krishnamurthy$^{1,3}$, 
        Anthony Tzes$^{1,2}$, and Farshad Khorrami$^{1,3}$% <-this % stops a space
\thanks{$^{1}$Center for Artificial Intelligence and Robotics  (CAIR), New York University Abu Dhabi, United Arab Emirates.}% <-this % stops a space
% <-this % stops a space
\thanks{$^{2}$New York University Abu Dhabi, Electrical Engineering, Abu Dhabi 129188, United Arab Emirates.}
\thanks{$^{3}$New York University, Electrical \& Computer Engineering Department, Brooklyn, NY 11201, USA.}
}

\maketitle

\author{IEEE Publication Technology,~\IEEEmembership{Staff,~IEEE,}
        % <-this % stops a space
\thanks{This paper was produced by the IEEE Publication Technology Group. They are in Piscataway, NJ.}% <-this % stops a space
\thanks{Manuscript received April 19, 2021; revised August 16, 2021.}}

%\author[CAIR]{Vinicius Mariano Gon\c{c}alves}\ead{vmg6973@nyu.edu},  
%\author[NYUNY]{Dimitris Chaikalis}\ead{dimitris.chaikalis@nyu.edu},
%\author[CAIR,NYUAD]{Anthony Tzes}\ead{anthony.tzes@nyu.edu},
%\author[CAIR,NYUNY]{Farshad Khorrami}\ead{khorrami@nyu.edu}

%\address[CAIR]{Center for Artificial Intelligence and Robotics  (CAIR), New York University Abu Dhabi, United Arab Emirates}  
%\address[NYUNY]{New York University, Electrical \& Computer Engineering Department, Brooklyn, NY 11201, USA}
%\address[NYUAD]{New York University Abu Dhabi, Electrical Engineering, Abu Dhabi 129188, United Arab Emirates}

\maketitle

\begin{abstract} 
Control Barrier Functions and Quadratic Programming are increasingly used for designing controllers that consider critical safety constraints. However, like Artificial Potential Fields, they can suffer from the stable spurious equilibrium point problem, which can result in the controller failing to reach the goal. To address this issue, we propose introducing circulation inequalities as a constraint. These inequalities force the system to explicitly circulate the obstacle region in configuration space, thus avoiding undesirable equilibria. We conduct a theoretical analysis of the proposed framework and demonstrate its efficacy through simulation studies. By mitigating spurious equilibria, our approach enhances the reliability of CBF-based controllers, making them more suitable for real-world applications.
\end{abstract} 

\section{Introduction}
 Control barrier functions \cite{8796030} (\textbf{CBF}s) have attracted much interest in recent years in the field of automatic control, since they provide a convenient incorporation of safety constraints as linear inequalities in the control input \cite{SafeTeleop_2018, panagou2019, GuaObsAvoid_2021, SafeCSyn_2021, CBFSingAvoid_2021, SafetyFromFast_2022, LearnBetterCBF_2022, Onboard_Safety_2022, explsol2022, Krstic2022}. This inequality can be inserted into a Quadratic Programming (\textbf{QP}) optimization problem, that is formulated based on two objectives:  making the resulting control action as close as possible to a \emph{nominal controller}, that solves the task (e.g., navigating towards a goal location) if we disregard all the constraints;  ensuring that all the safety constraints are satisfied. Solving this QP optimization problem provides a control input that is able to preserve safety while driving the robot towards the goal. 

 In the QP formulation, the constraints are linear in the  control input provided that we assume an affine dynamics for the system. This linearity appears since we use \emph{local information} of the problem, like distance function and Lyapunov function gradients. This locality of the approach induces \emph{spurious equilibrium points} into the system, and so the closed loop system can become stuck in such a point. Even worse, these spurious equilibrium points can even be \emph{stable}~\cite{9125992}.

 This problem also appears in other approaches that use local information to navigate. A quintessential example is using \emph{artificial potential functions} (\textbf{APF}) \cite{163777}, where a combination of gradients of attractive and repulsive potentials are employed for navigation. The QP-CBF formulation is qualitatively similar to the APF. In fact, the Karush-Kuhn-Tucker conditions (\textbf{KKT}),  the fact that many of the CBF inequalities are generated using gradients of ``obstacle functions,'' and the fact that the nominal controller is usually the negative gradient of a potential function, implies that the CBF-QP controller can be seen as a sum of attractive and repulsive vectors with a  choice of weights that comes from the optimization procedure.

Considering this local equilibrium issue, there are works that aim to modify the traditional QP-CBFs/APF approaches to circumvent it. For example, \cite{9125992} introduces a new CBF inequality into the original CBF formulation that can, in the case that there is a single convex obstacle, remove the equilibrium points in the boundary of the obstacle. In the proposed approach in this paper, the spurious equilibrium point problem is mitigated by introducing an inequality, which we call \emph{a circulation inequality}. This inequality forces the configuration to circulate the obstacle when near it. This is motivated by the observation that QP-CBF approaches often fail in large/complex obstacles, specially when the nominal controller is defined in terms of the gradient of a simple potential function as, for example, the squared distance between the configuration and the goal in the configuration space.  This is because the controller seeks to decrease this potential function, and for large obstacles this may be impossible to do in all situations: sometimes the configuration must move around (i.e., circulate) the obstacle, resulting in temporary increase of this function so as to enable eventual reaching of the goal at a later time. 

The idea of inducing circulation of obstacles appears in other works as well. Early approaches for circulating obstacles appear in the \emph{bug algorithm} \cite{6280933}. The proposed inequality can be seen as a low-level implementation of this circulation aspect, although the proposed controller does not realize exactly the bug algorithm or any of its variants. In order to avoid obstacles, \cite{giroscopic2003} proposes the addition of a circulation term, called ``gyroscopic force'', and \cite{ataka2018, magnetic2022} propose the addition of a circulation term based on magnetic fields. Furthermore, in  \cite{GAO2023104291}, the authors propose a modification of the classic attractive/repulsive potential fields approach so as to improve its ``efficiency'' while retaining the Lyapunov stability analysis from traditional potential fields. The approach is motivated by the observation that if the attractive and repulsive vectors conflict with each other - i.e., if the cosine of the angle between them is less than 0 - then, within a certain range of angles, the repulsive vector can be replaced by a different vector that is orthogonal to the attractive vector without violating the Lyapunov inequalities for the closed loop system. This modification of the attractive/repulsive potential fields approach allows the robot to move faster towards the goal since the impact of the repulsive potential is reduced in the system. This modification of making (in some circumstances) the repulsive and attractive fields orthogonal can be seen as a way to generate a circulation behavior in the controller. However, although this technique can remove some spurious equilibrium points in some situations, it cannot completely solve the problem. The ``modification of the repulsive/attractive vectors'' philosophy is also employed in this work, since the proposed optimization formulation is conceptually similar to APF approaches. The difference is that the modification is introduced in this paper by adding a new constraint into the optimization problem. This proposed strategy allows us to remove the spurious equilibrium points in some circumstances, as shown in Figure~\ref{fig:comparison}.
\vspace*{-8mm}
\begin{figure}[htbp]
\subfloat[\centering Without circulation]{\includegraphics[width=4.53cm,trim={6.5cm 9.5cm 5cm 8cm},clip]{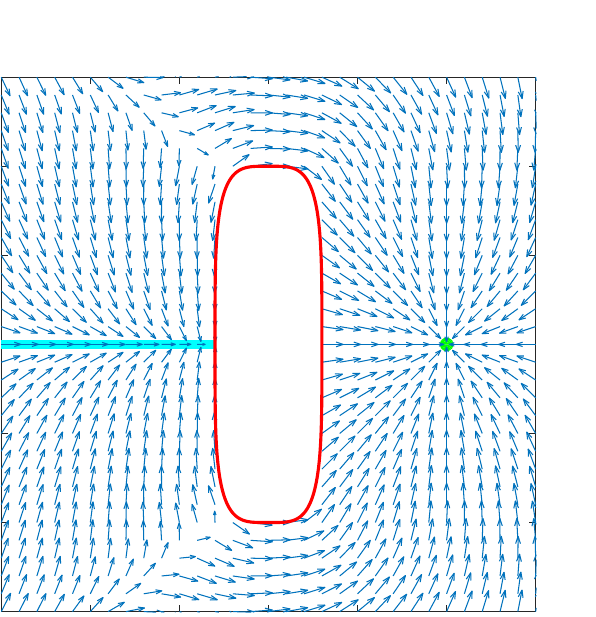}}
\subfloat[\centering With circulation]{\includegraphics[width=4.08cm,trim={6.5cm 9.5cm 6cm 8cm},clip]{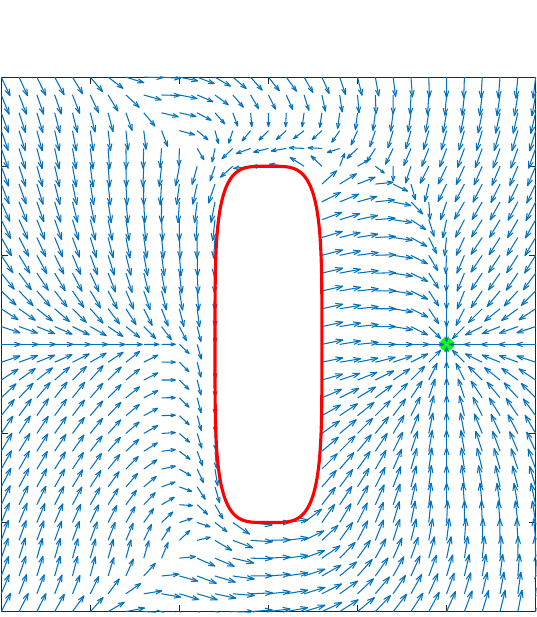}}
\centering
\caption{Vector field generated by traditional CBF-QP formulations \eqref{eq:formulation} (left) and  our formulation
\eqref{eq:formulation2} (right), with a nominal controller that points towards the goal. Note that in the original formulation, there is a set of points (at the left of the obstacle) such that if we start from this set, we converge to  spurious equilibria near the boundary of the obstacle instead of converging towards the target.}
\label{fig:comparison}
\end{figure}

The problem considered in \cite{9992435}, which is to track a path while avoiding obstacles,  is similar to the problem addressed in this paper. To implement obstacle avoidance, a ``circulation mode'', based on the controller proposed in \cite{goncalves2010}, is initiated when the robot is near the obstacle until it can return to navigating towards the target. This switching between circulation and goal-directed navigation modes is done by monitoring the sign of a variable and can produce discontinuities in the vector field. Our proposal is similar in the sense that we also enforce circulation when near an obstacle, but the enforcing of this behavior is achieved by imposing a circulation constraint, which gradually increases in its effect as the robot approaches the obstacle, instead of switching the vector fields in a discontinuous fashion. Since our formulation is based on a strictly convex optimization problem with continuous data, our proposed construction can be shown to have continuity properties using classical theorems on QPs.

 In this paper, we extend the QP-CBF formulations with a circulation constraint, providing certain theoretical results (as characterization of the equilibrium points, feasibility, and continuity). The major challenge is related to how to \emph{define} circulation in higher dimensions. In order to implement the constraint, we need to specify precisely \emph{how} to circulate the obstacles.  Clockwise or counter-clockwise can be used in planar situations, but when we consider higher dimensions the situation is more complex, in which topology may forbid us to defining circulation in a continuous way. 

 \noindent{\em Paper organization:} The notations are summarized in Section~\ref{sec:notations}. The proposed CCBF-QP (Circulation Control Barrier Function Quadratic Program) approach and its properties are discussed in Section~\ref{sec:formulation}. Practical aspects in application of the proposed approach are discussed in Section~\ref{sec:practical}. Simulation studies of the approach are presented in Section~\ref{sec:simulations}. Finally, concluding remarks are summarized in Section~\ref{sec:conclusion}.
\section{Mathematical Notation}
\label{sec:notations}
We denote by $\mathbb{R}^+~(\mathbb{R}^-)$ the set of nonnegative~(nonpositive) reals. All vectors are column vectors unless stated otherwise, and for a vector $v$, $v^T$ denotes its transpose. If $f: \mathbb{R}^n \mapsto \mathbb{R}$, $\nabla f$ is its gradient, written as a column vector. For a closed set $\mathcal{C}$, we denote its boundary by $\partial \mathcal{C}$ and $\textsl{Int}(\mathcal{C}) \triangleq \mathcal{C}-\partial \mathcal{C}$ its interior. We denote the identity matrix of appropriate order by $I$. A square matrix is \emph{skew-symmetric} if $\Omega^T=-\Omega$ and \emph{orthonormal} if $\Omega^T\Omega=I$.

\section{The CCBF-QP formulation and its properties}
\label{sec:formulation}
\subsection{The CCBF-QP formulation}

We are interested in defining a controller that reaches a goal in the configuration space while avoiding obstacles. For this purpose, we will state certain assumptions (and related definitions) that will allow us to describe the problem and our proposed solution formally.

\begin{assumption} \label{assumption:1} \ \ \
We assume the following:
\begin{enumerate}
\item \label{item:assump1-1} Our configuration space  $\mathcal{Q}$ is an \emph{even-dimensional} Euclidean space, that is, $\mathcal{Q} \subset \mathbb{R}^n$, $n$ even. Note that $\mathcal{Q}$ can include forbidden configurations.% (e.g that implies collisions)

\item \label{item:assump1-3} We consider a single $n$-dimensional 
obstacle $\mathcal{C} \subset \mathbb{R}^n$ in the configuration space that is modeled as a closed set. Furthermore, we assume a function $D: \mathbb{R}^n \mapsto \mathbb{R}^+$ that computes a metric of ``closeness'' from the point $q \in \mathcal{Q}$ and $\mathcal{C}$. $D$ is not required to be the Euclidean distance between $q$ and $\mathcal{C}$, but 
\begin{equation}
D(q) > 0 ~\Rightarrow q \not \in \mathcal{C}.
\end{equation}
Furthermore, $D$ is assumed to be differentiable (i.e., $\nabla D(q)$ is continuous) everywhere in $\mathbb{R}^n - \textsl{Int}(\mathcal{C})$ and $\nabla D(q) \not= 0$ almost everywhere. Define the set $\mathcal{W} \triangleq \{q \in \mathbb{R}^n \ | \ q \in \mathbb{R}^n - \textsl{Int}(\mathcal{C}) \ , \ \nabla D(q) \not=0\}$. For $q \in \mathcal{W}$, we also define the \emph{normal vector} $N(q) \triangleq \nabla D(q)/\|\nabla D(q)\|$ 
.
\item \label{item:assump1-4} A skew-symmetric orthonormal matrix is defined such that $\Omega \in \mathbb{R}^{n \times n}$. %, i.e, a matrix $\Omega$ such that $\Omega^T=-\Omega$ and $\Omega^T\Omega = I$. 
If $q \in \mathcal{W}$, we define the \emph{tangent vector} as $T(q) \triangleq \Omega N(q)$.

\item \label{item:assump1-5} We aim to achieve a target configuration $q_{g} \in \mathcal{W}$, and that we have a \emph{nominal} controller $u_d: \mathbb{R}^n \mapsto \mathbb{R}^n$ that achieves this task, as the negative gradient of a positive definite differentiable function $V(q)$ such that $V$ and $\nabla V$ vanish only at $q=q_{g}$. Furthermore, we assume that $\|\nabla V(q)\|$ is upper bounded by all $q$, i.e., $\displaystyle \max_{q \in \mathbb{R}^n} \|\nabla V(q)\|$ is finite.

\item \label{item:assump1-8} We have a set of configuration-dependent constraints to our control input, written as $\mathcal{U}(q) \triangleq \{\mu \ | \ A(q) \mu \geq b(q)\}$, in which $A: \mathbb{R}^n \mapsto \mathbb{R}^{l \times n}$ and $b: \mathbb{R}^n \mapsto \mathbb{R}^{l}$ are continuous functions of $q$. Furthermore, $\forall q$, let the set $\mathcal{U}(q)$ contain a ball of radius $r > 0$ centered at the origin, i.e., if $\|\mu\| \leq r$, then $A \mu \geq b$. 

\item \label{item:assump1-6} We have a continuous function $\alpha: \mathbb{R}^+ \mapsto \mathbb{R}^-$ such that $\alpha(0)=0$ and $\alpha(D)<0$ for $D < 0$ (e.g., $\alpha(D) = -\eta D$ for a constant $\eta > 0$) .

\item \label{item:assump1-7} We have a continuous function $\beta: \mathbb{R}^+ \mapsto \mathbb{R}$ with the following properties: $\beta$ is decreasing, $\beta(0) > 0$, and $\beta(D) \rightarrow -\infty$ as $D \rightarrow \infty$ (e.g., $\beta(D) = a - bD$ for positive constants $a,b$). Furthermore, $\beta(D(q_g)) < 0$ and $\beta(0) < r$.

\item \label{item:assump1-2} Our system's dynamics is $\dot{q} = u$%, i.e, velocity control
.
\end{enumerate}
\end{assumption}

The most striking assumption is related to the even-dimensional assumption in Assumption \ref{assumption:1}-\ref{item:assump1-1}. It turns out that even-dimensional spaces are more amenable for defining the circulation of obstacles than odd-dimensional spaces. This is a topological issue: the \emph{hairy ball theorem} \cite{10.2307/2320587} establishes that there is no way to define a continuous and non-vanishing circulation field (i.e., tangent to the surface) on the surfaces of odd-dimensional balls. Thus, since $N$ is in the surface of a $n$-dimensional ball (thanks to $\|N\|=1$ and $N   \in \mathbb{R}^n$), creating a continuous non-vanishing function that maps every possible $N$ to a vector orthogonal to it would be tantamount to generating a continuous non-vanishing field at the surface of the $n$-sphere, which is forbidden by the hairy ball theorem when $n$ is odd. As a consequence of this result, there is no such matrix $\Omega$ as defined in Assumption \ref{assumption:1}-\ref{item:assump1-4} if $n$ is odd\footnote{Another way to see it without the hairy ball theorem is that if $n$ is odd and $\Omega$ is skew-symmetric, $\det(\Omega^T) = \det(-\Omega)$ implies $\det(\Omega)=(-1)^n\det(\Omega)$ and thus $\det(\Omega)=0$. Thus, odd-dimensional skew-symmetric matrices are non-invertible and consequently cannot be orthonormal.}, because $F(N) = \Omega N$ would define a continuous and non-vanishing tangential field in the surface of an odd-dimensional sphere. In Subsection \ref{subs:odddimensions}, we discuss a workaround for the case when $n$ is odd.

If we want to reach the target $q_{g}$ while avoiding the obstacle, a common approach is to formulate a \emph{minimally invasive QP using CBFs} and compute the control input $u$ as:
\begin{eqnarray}
\label{eq:formulation}
u(q) = && \arg \min_\mu \ \|\mu-u_d(q)\|^2  \nonumber \\
\mbox{such that}&& N(q)^T \mu \geq \alpha(D(q)) \nonumber  \\
&& A(q) \mu \geq b(q)
\end{eqnarray}

\noindent for $q \in \mathcal{W}$ (we leave $u(q)$ undefined if $q \not \in \mathcal{W}$). The QP tries to minimally modify the nominal controller $u_d(q)$ while imposing the CBF constraint $\dot{D} = \nabla D(q)^T \mu \geq  \|\nabla D(q)\| \alpha(D(q)) \geq 0$ that guarantees that the system will not enter the forbidden set $\mathcal{C}$.

Unfortunately, the formulation in \eqref{eq:formulation} is prone to have stable equilibrium points besides the only desired one $q=q_{g}$. The result of \eqref{eq:formulation} is essentially a sum of potential fields, which is known for having spurious equilibrium points, dependent on the repulsive potential $N(q)$. Indeed, if we disregard the constraints $A(q)u \geq b(q)$, \eqref{eq:formulation} can be solved analytically:
\begin{equation*}
u(q) = -\nabla V(q) + N(q) \lambda(q)
\end{equation*}
\noindent in which $\lambda(q) = \max\big(0,N(q)^T \nabla V(q)+\alpha(D(q))\big)$, the dual variable for the normal constraint, is a nonnegative configuration-dependent weight for the repulsive potential.

These spurious equilibrium points appear near the surface of the obstacle when the controller cannot force the configuration to circulate the obstacle in order to escape from it. Therefore, in this paper, we augment \eqref{eq:formulation} with an inequality that forces the configuration to circulate the obstacle when in its neighborhood by including the inequality:
\begin{equation*}
T(q)^T u \geq \beta(D(q))
\end{equation*}

\noindent (in which $\beta$ was defined in Assumption \ref{assumption:1} - \ref{item:assump1-7}) that forces the following behavior. When near obstacles (that is, when $\beta > 0$), we have a positive projection into the tangent vector of the obstacle, which means that we should circulate it. If $D$ is large enough, $\beta$ becomes very negative and thus the constraint is trivialized%, that is, it does not impose any constraint on the system
. This is specially true for $q=q_g$: we do not want to force any circulation when the goal is achieved %(because we want to stop there), 
and this is why Assumption \ref{assumption:1} - \ref{item:assump1-7} has $\beta(D(q_g))<0$. Nevertheless, after the inclusion of this inequality into \eqref{eq:formulation}, we have the \textbf{CCBF-QP} (Circulation Control Barrier Function Quadratic Program), including the validity of all terms in Assumption \ref{assumption:1}:
\begin{eqnarray}
\label{eq:formulation2}
u(q) = && \arg \min_{\mu} \ \|\mu-u_d(q)\|^2 \nonumber \\
\mbox{such that} && N(q)^T \mu \geq \alpha(D(q)) \nonumber  \\
&& T(q)^T \mu \geq \beta(D(q)) \nonumber  \\
&& A(q) \mu \geq b(q).
\end{eqnarray}

Figure \ref{fig:comparison} makes the comparison in the computed vector fields between \eqref{eq:formulation} and \eqref{eq:formulation2}, disregarding the constraints $A(q)\mu \geq b(q)$, for a smooth rectangle-like obstacle and goal $p_{\textsl{goal}} = [1 \ 0]^T$ (green dot in the two pictures). We can see that without the circulation term, there is a spurious line of attractive points (in cyan), with a very large attractive basin to the left of the obstacle. If we include the circulation terms, there are no spurious attractive points anymore%: the only attractive point is the goal
.

Next, we establish some properties for this formulation.

\subsection{Feasibility of CCBF-QP}

One important question is whether the CCBF-QP formulation always admits a feasible solution. The following result answers this question positively.

\begin{theorem} \label{prop:feasibility} Let $K$  be such that $\beta(0) < K < r$ (such $K$ always exists, since $\beta(0)<r$, see Assumption \ref{assumption:1}-\ref{item:assump1-7}). Then, if $q \in \mathcal{W}$, the CCBF-QP formulation  always has a feasible solution: $\mu_f \triangleq K T(q)$. Furthermore, this particular solution lies in the interior of the feasible set, achieving all inequalities without equality.
\end{theorem}
\begin{proof} We shall verify all of the three constraints in \eqref{eq:formulation2}, and emphasize that the inequalities are achieved without equality.

\begin{itemize}
    \item For $N(q)^T \mu_f \geq \alpha(D(q))$, we note that $N(q)^TT(q) = 0$ and $\alpha(D(q)) < 0 $ if $q \not \in \mathcal{C}$. 
    \item For $T(q)^T \mu_f \geq \beta(D(q))$, we note that  $T(q)^TT(q) = \|T(q)\|^2 = 1$ and that $K> \beta(D(q))$ since $\beta$ being decreasing (Assumption \ref{assumption:1}-\ref{item:assump1-7}) implies that the maximum value for $\beta$ occurs at $D=0$ and thus $K > \beta(0) \geq \beta(D(q))$.
    \item  For the constraint $A\mu \geq b$, we use the fact that $\|\mu\| \leq r$ is a sufficient condition for this inequality to hold (Assumption \ref{assumption:1}-\ref{item:assump1-8}). Now $\|  K T(q)\| = K < r$. The inequality being strict means that $\mu_f$ lies strictly inside the set $A \mu \geq b$, i.e., $A \mu_f > b$.
\end{itemize}  
\end{proof}

This implies that CCBF-QP allows as a feasible action for the system to move tangent to its current contour surface of $D(q)$ (in the direction specified by $T(q)$).

\subsection{Continuity of CCBF-QP}

We now prove the continuity of $u(q)$.

\begin{theorem} If $u(q)$ comes from a CCBF-QP, it is a continuous function of $q$ for any $q \in \mathcal{W}$.
\end{theorem}
\begin{proof}
    This is a consequence of a result \cite{Lee2005} that states the following: the single minimizer of a strictly convex quadratic program $\displaystyle \min_{\mu} \mu^TH\mu/2 + f^T\mu \ \mbox{such that} \ G\mu \geq g$ is continuous on the parameters $\{H,f,G,g\}$ as long as the constraint $G \mu \geq g$ is \emph{regular}, which means that there exists $\mu_f$ such that $G \mu_f > g$ (i.e., the interior is non-empty). 

    In \eqref{eq:formulation}, $H = 2I$, $f = 2\nabla V(q)$, $G=[N(q) \ T(q) \ A(q)^T]^T $, and $g(q) = [\alpha(D(q)) \ \beta(D(q)) \ b(q)^T]^T$. The regularity is a corollary of Theorem \ref{prop:feasibility}, that states that $\mu_f$ (defined in that theorem) satisfies the inequalities with all inequalities being strict. Now, the set of Assumptions \ref{assumption:1} guarantees that $f(q), G(q)$, and $g(q)$ are continuous functions of $q$ for $q \in \mathcal{W}$, and the proof is complete.
\end{proof}

\subsection{CCBF-QP Equilibrium Points}
We will characterize the subset of $\mathcal{W}$  in which the output of  CCBF-QP  is the null vector, that is, the \emph{equilibrium points} of the dynamical system $\dot{q} = u(q)$.

\begin{theorem}  \label{prop:equilibrium} Considering the set $q \in \mathcal{W}$, in the CCBF-QP formulation, $u(q) =0$ occurs only in the two mutually exclusive conditions:

\begin{enumerate}
\item $q = q_{g}$;
\item $\beta(D(q))=0$ and $-u_d(q)$ is positive parallel to $T(q)$, that is, there exists a positive scalar $\lambda_T$ for which $u_d+\lambda_T T=0$.
\end{enumerate}
\end{theorem}

\begin{proof} If $\mu=0$ is optimal, the KKT conditions, for the dual variables $\lambda_N, \lambda_T \in \mathbb{R}$ and $\lambda_A \in \mathbb{R}^{l}$, are \cite{convexanalysis}:
\begin{eqnarray}
\label{eq:kkt}
&& (i) : \ u_d + N \lambda_N + T \lambda_T + A^T \lambda_A = 0; \nonumber \\
&& (ii) : \  \lambda_N\alpha = 0 \ , \ (iii) : \ \lambda_T\beta = 0 \ , \ (iv) : \  \lambda_A^Tb = 0; \nonumber \\
&&  (v) : \ \lambda_N \geq 0 \ , \ (vi) : \ \lambda_T \geq 0 \ , \ (vii) : \ \lambda_A \geq 0; \nonumber \\
&& (viii) : \ 0 \geq \alpha \ , (ix) : \ 0 \geq \beta \ , \ (x) :  0 \geq b.
\end{eqnarray}

The first important fact is that it is not possible for these conditions to hold if $q \in \mathcal{C}$. Indeed, in this case, $\beta(D(q)) > 0$, which contradicts ($ix$).

Thus, if $q \not \in \mathcal{C}$, we have that $\alpha \not= 0$ (Assumption \ref{assumption:1}-\ref{item:assump1-6}). Consequently, we get from ($ii$) that $\lambda_N = 0$. Furthermore, since the polytope $A \mu \geq b$ contains a ball of non-null radius $r$ (Assumption \ref{assumption:1}-\ref{item:assump1-8}), $b$ is a strictly negative vector (since $\mu=0$ should be feasible and in the interior of the polytope), and then, from ($iv$) and ($vii$), we have that $\lambda_A=0$. Hence,
\begin{equation}
\label{eq:kkt2}
(i) : \ u_d + T \lambda_T = 0 \ , \ (iii) : \ \lambda_T \beta=0 \ , \ (vi) : \ \lambda_T \geq 0 .
\end{equation}
From ($iii$), either $\lambda_T = 0$ or $\beta=0$. If $\lambda_T=0$, we have $u_d(q)=0$. Due to Assumption \ref{assumption:1}-\ref{item:assump1-5}, this implies $q=q_{g}$. If $\beta=0$, we have $u_d+T\lambda_T=0$ with $\lambda_T \geq 0$. Note that since $\beta(D(q_{g})) \not =0$ (Assumption \ref{assumption:1}-\ref{item:assump1-7}), the two sets of vanishing points are mutually exclusive.
\end{proof}

\section{Practical aspects}
\label{sec:practical}
\subsection{The function D}
\label{subs:thefunD}
The scalar function $D$ must encompass all the system constraints, codified as an obstacle $\mathcal{C}$ in the configuration space. However, in practice, it is easier to describe these constraints using $m$ scalar functions as $F_i(q) \geq 0$, that is, $\mathcal{C} = \{q \in \mathbb{R}^n \ | \ F_i(q) \geq 0 \ i = 1,2,..,m\}$. In this case, one could use as function $D$ the minimum of these $F_i$, since clearly $D > 0$ implies that $q \not \in \mathcal{C}$. However, even if $F_i$ is everywhere differentiable, the function $D$ defined in this way can have many non-differentiable points. More specifically, whenever we have a $q$ in which the minimum is achieved for at least two different functions $F_i$, $D$ is non-differentiable unless the respective gradients $\nabla F_i$ are equal at this $q$, which is a very strong condition. Being non-differentiable brings several practical problems because the controller will be discontinuous at that point.

One idea to remove this problem is to use an approximate version of the minimum function. The \emph{softmin} function \cite{Gao2017OnTP} is a good candidate, defined (with an averaging term) as $ \min^h(g_1,g_2,...,g_m) \triangleq -h \ln(\frac{1}{m}\sum_k \exp(-g_k/h))$. This function is always differentiable in the arguments $g_i$ and approximates the minimum when $h$ is positive and  close to $0$. However, we can show that $ \min^h(g_1,g_2,...,g_m) \geq \min(g_1,g_2,...,g_m)$, and thus if we choose $D(q) =  \min^h_i F_i(q)$, $D > 0$ may not imply $F_i(q)>0$ for all $i$. In this case, we introduce a margin $\delta  > 0$ and choose $D(q) \triangleq  \min^h_i F_i(q)-\delta$.  Another important aspect is that a naive implementation of the softmin function can easily generate numerical precision problems, since some of the terms $\exp(-g_i/h)$ can be too large (so we have a $\ln(\infty)$ situation) or all of them too small (so we have a $\ln(0)$ situation). In order to avoid this, we can use the fact that $\min^h_i g_i = \min_i g_i + \min^h_i (g_i-\min_i g_i)$. The terms $\hat{g}_i \triangleq g_i-\min_i g_i$ are always positive and at least one of them is $0$, which implies that the terms $\exp(-\hat{g}_i/h)$ will be always less than or equal to one and at least one of them is $1$. This eliminates this numerical issue.

\subsection{Generating matrices $\Omega$}
\label{subs:genOmega}

In order to implement the algorithm, we need skew-symmetric orthonormal matrices, i.e, $\Omega \in \mathbb{R}^{n \times n}$ such that $\Omega^T = -\Omega = \Omega^{-1}$. As mentioned before, they exist only in even-dimensional spaces. A different choice of $\Omega$ implies a different way to circulate the obstacle.

If $n=2$, there are only two such matrices:

\begin{equation}
    \Omega_1 = \left[\begin{array}{cc} 0 & -1\\ 1 & \ 0\end{array}\right] \ , \ \Omega_2 =  -\Omega_1 = \left[\begin{array}{cc} \ 0 & 1\\ -1 & 0\end{array}\right]
\end{equation}

\noindent corresponding to counter-clockwise and clockwise rotations, respectively. But for $n > 2$ (and even), infinite such matrices exist. A simple way to create these matrices is to divide the set $\{1,2,...,n\}$ into $n/2$ pairs $(i_k,j_k)$, $k=1,2,...,n/2$. We then start from the $n \times n$ identity matrix, and multiply each $i_k^{th}$ row by $-1$. We then swap the $i_k^{th}$ row with the $j_k^{th}$ row, and we can check by inspection that the resulting matrix is skew-symmetric and orthonormal. For example, for $n=4$, we can have the pairs $(1,3)$ and $(2,4)$ and then:

\begin{equation}
    \Omega = \left[\begin{array}{cccc}
    \ 0 & \ 0 & 1 & 0 \\
    \ 0 & \ 0 & 0 & 1 \\
    -1 & \ 0 & 0 & 0 \\
    \ 0 & -1 & 0 & 0
    \end{array}\right].
\end{equation}

We can create $n!/(n/2)!$ such matrices, but for even numbers $n > 2$, there are other ways. One simple way is as follows: it is easy to show that if $\Omega \in \mathbb{R}^{n \times n}$ is skew-symmetric orthonormal and $Q \in \mathbb{R}^{n \times n}$ is an orthonormal matrix, then $\hat{\Omega} = Q \Omega  Q^{T}$ is also a skew-symmetric orthonormal matrix. We can use this fact to generate new skew-symmetric orthogonal matrices by constructing $\Omega$ by the simple pairing index procedure described above,  constructing an orthonormal matrix $Q$ and then calculating $\hat{\Omega} = Q \Omega  Q^{T}$. To generate orthonormal matrices in $\mathbb{R}^{n \times n}$, we use the fact that if $A$ is skew-symmetric and $I$ is the $n \times n$ identity matrix, $Q = (I-A)^{-1}(I+A)$ is orthonormal. We can generate a skew-symmetric matrix by using any matrix $B \in \mathbb{R}^n$ and computing $A=B-B^T$.

\subsection{Odd dimension extension}
\label{subs:odddimensions}

In Assumption \ref{assumption:1}-\ref{item:assump1-1}, we consider even $n$. If $n$ is odd, there is no matrix $\Omega \in \mathbb{R}^{n \times n}$ that is skew-symmetric and orthonormal. One alternative for handling the odd-dimensional case is to drop one of these properties. On the one hand, being skew-symmetric is important for guaranteeing that $T = \Omega N$ is orthogonal to $N$ (for any $N$), which guarantees that $T$ indeed induces circulations. On the other hand, being orthonormal, is not as important as $\Omega$ being invertible, which is also incompatible with skew-symmetry when $n$ is odd. Invertibility is important because it guarantees that $T = \Omega N$ only vanishes when $N$ vanishes, i.e., when $q \not \in \mathcal{W}$. If $T$ vanishes when $\beta > 0$, the circulation constraint $T^T \mu \geq \beta$ cannot be satisfied and the QP is infeasible, rendering the controller undefined even in some points inside $\mathcal{W}$.  

This constraint in $\Omega$ for $n$ odd can be relaxed if we assume that $\Omega$ is skew-symmetric but is orthonormal only in an $n-1$ dimensional subspace of $\mathbb{R}^n$, being singular in the complementary $1$-dimensional space. Such matrices can be created using a similar construction as the one described in Subsection \ref{subs:genOmega}. However, if $n$ is odd, the index $i$ that has no pair must have its respective row in the identity matrix set to zero.
\section{Simulation Studies}
\label{sec:simulations}
\subsection{4-DoF planar manipulator}
In this simulation, we have a planar 4DOF revolute manipulator and rectangle-shaped links of size $0.70$m$\times0.15$m. The configuration $q \in \mathbb{R}^4$ corresponds to the four joint positions. It must go from an initial horizontal configuration to a vertical configuration while avoiding a circular obstacle with radius 0.5m, centered at $x$=1.5m and $y=1.3$m. %(see Figure \ref{fig:manippath}). 

\begin{figure}
\includegraphics[width=4cm,trim={6cm 5cm 6cm 5cm},clip]{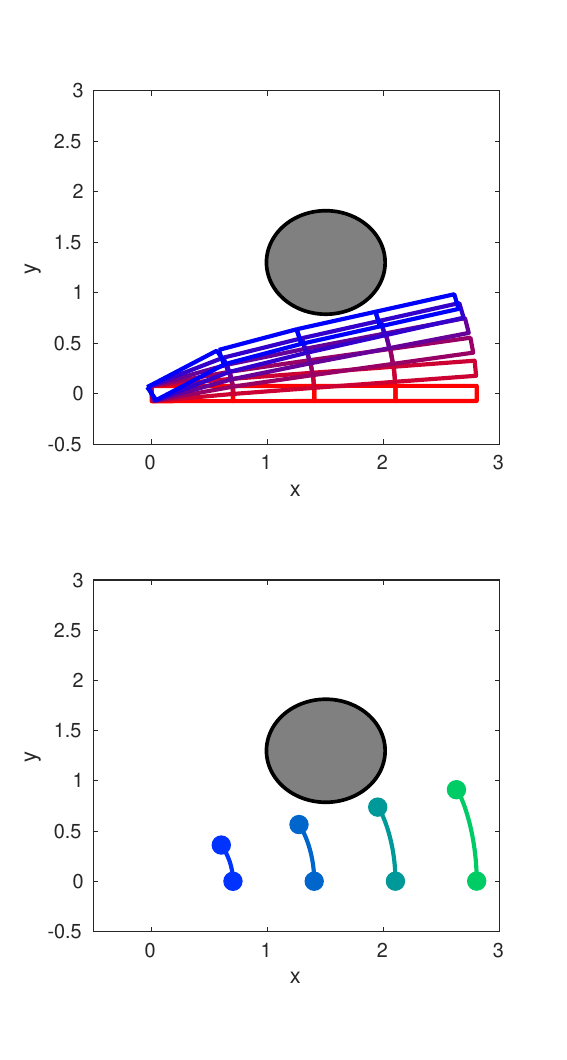}
\includegraphics[width=4cm,trim={6cm 5cm 6cm 5cm},clip]{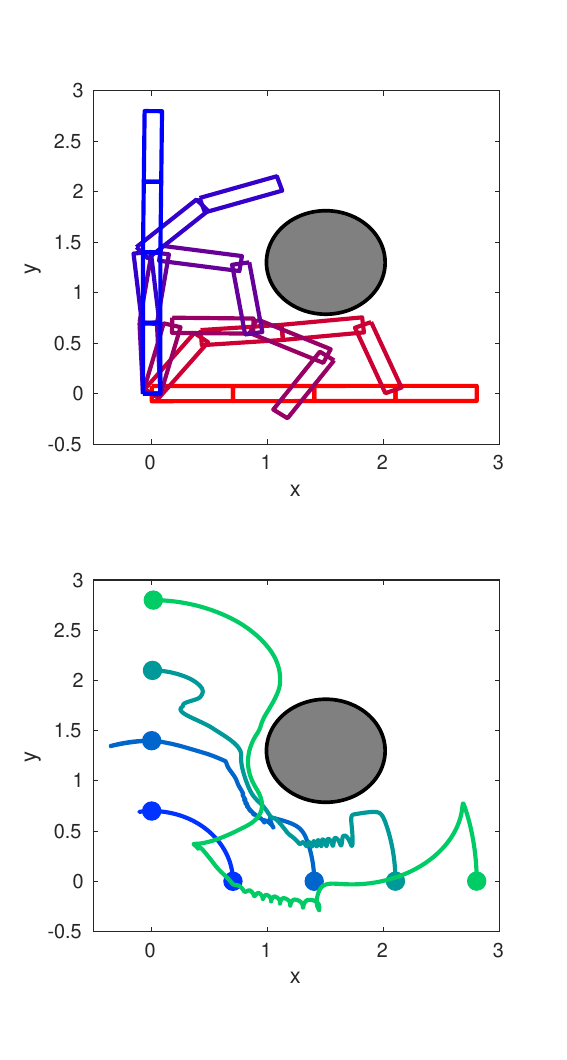}
\centering
\caption{Manipulator trajectory snapshots without (with) the circulation constraint at left (right) column.
%Top row: Snapshots of the trajectory of the manipulator from the starting configuration (red) to the target configuration (blue).  Bottom row: path made by the junction points of the links and the tip of the manipulator.
}
\label{fig:manippath}
\end{figure}

 We introduce functions $F_i$ to model the constraints and obstacles in configuration space, as described in Subsection \ref{subs:thefunD}. The first constraint is the collision avoidance between the links and the obstacle (i.e., circle). We consider functions $F_i(q)$, $i=1,2,3,4$ that compute the squared distance (measured in meters square) between each rectangular link and the circle. We also consider a joint limit of $\pm 360^o$ for the first joint and $\pm 120^o$ for all the other joints. These constraints also generate obstacles in the configuration space, which are modeled using functions $F_i(q) = 0.02(q_i-q_{\min,i})$ and $F_i(q) = 0.02(q_{\max,i}-q_i)$ ($q_i$ measured in degrees) for minimum and maximum joint limits, respectively. We thus have another 8 functions $F_i$. For the softmin, we used $h=0.04$ and a margin of $\delta=0.15$.
We used  $\alpha(D) = -0.5D$, $\beta(D) = 0.6-2.25D$, and $V(q) = \sum_{i=1}^4 0.4|q_i-q_{g,i}|^{1.5}$ . We also use the constraints of the form $A \mu \geq b$ to impose a maximum absolute value of $1$ to each control input. The employed $\Omega$ matrix is
\begin{equation}
    \Omega=\left[\begin{array}{cccc}
   \ 0  &  \ \frac{\sqrt{3}}{3}  & -\frac{\sqrt{3}}{3}  &  \frac{\sqrt{3}}{3} \\
   -\frac{\sqrt{3}}{3}  &  \ 0 &  \ \frac{\sqrt{3}}{3}  &  \frac{\sqrt{3}}{3} \\
   \ \frac{\sqrt{3}}{3}  & -\frac{\sqrt{3}}{3}  & \ 0  &  \frac{\sqrt{3}}{3} \\
   -\frac{\sqrt{3}}{3}  & -\frac{\sqrt{3}}{3}  & -\frac{\sqrt{3}}{3}  &  0 \\
   \end{array}\right].
\end{equation}

\begin{figure}
\includegraphics[width=9cm,trim={5.2cm 5cm 5cm 5cm},clip]{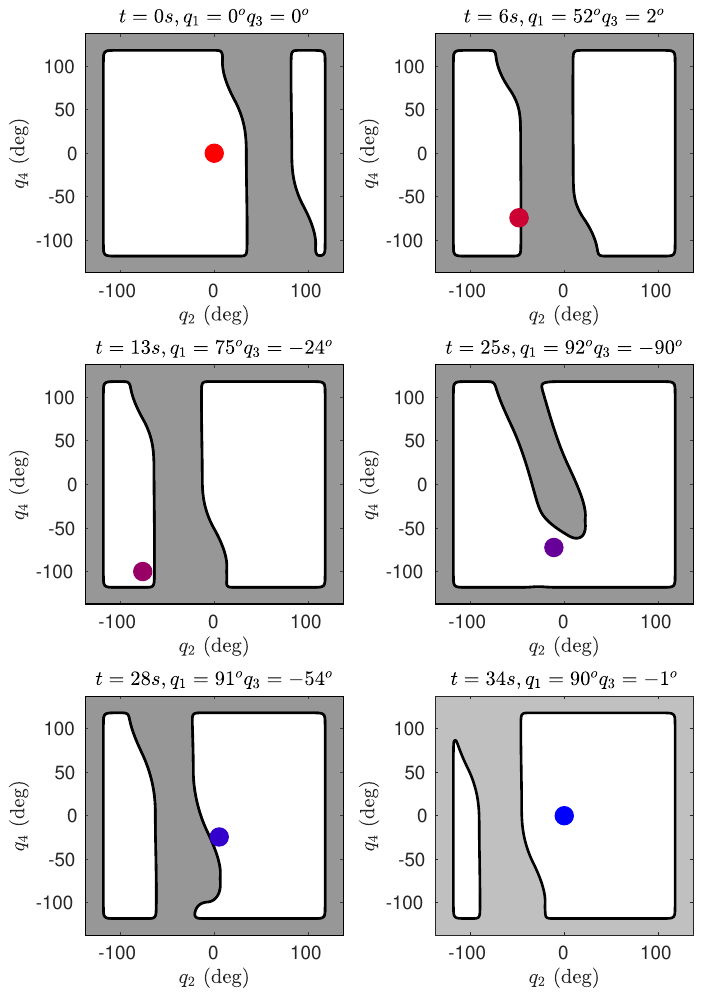}
\centering
\caption{Snapshots of the trajectory of the manipulator in the \emph{configuration space}.% The same timesteps shown in Figure~\ref{fig:manippath} using identical color code.
}
\label{fig:manippathconfig}
\end{figure}

Figure \ref{fig:manippath}-(Top-Right) shows the path made by the manipulator in six different timesteps, Figure \ref{fig:manippath}-(Bottom-Right) shows the path made by the junction points of the links/tip of the manipulator, and Figure \ref{fig:Dgraphmanip} shows the evolution of the distance function, which is always greater than 0. It is important to note that without the circulation constraint, the system converges to a spurious equilibrium point when the manipulator touches the circle from below. The snapshots for the simulation without the circulation constraint can be seen in the two figures at the Left column of Figure \ref{fig:manippath}. It is also important to mention that, although the obstacle in the 2D space is very simple (a circle), the resulting obstacle $\mathcal{C}$ in the configuration space is very complex. To illustrate this claim, Figure \ref{fig:manippathconfig} shows snapshots of the obstacle in the \emph{configuration space} during the same six timesteps as Figure \ref{fig:manippath}. Since the configuration space is 4-dimensional, we show 2-dimensional slices of this space by keeping the same configurations $q_1$, $q_3$ but changing $q_2$ and $q_4$. The obstacle region is in gray.

\begin{figure}
\includegraphics[width=8cm,trim={3.8cm 11cm 4cm 11cm},clip]{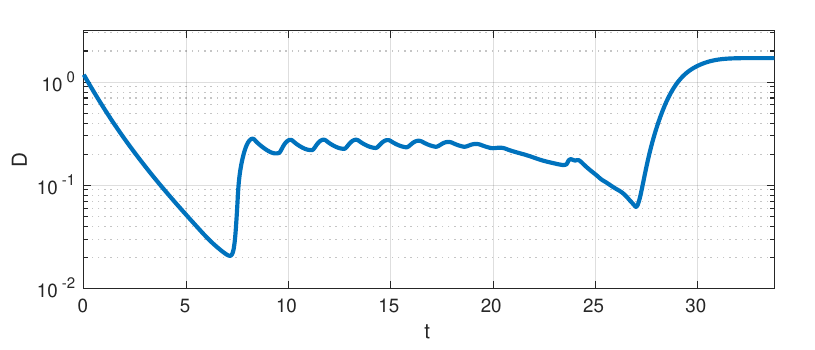}
\centering
\caption{Evolution of the distance function $D$. %The CBF constraint guarantees that it is always greater than 0.
}
\label{fig:Dgraphmanip}
\end{figure}
\subsection{Multiple robot 3D-motion ($n=15$ case)}
In this simulation, we have 5 flying robots for which we can control their linear velocities.  The configuration $q$ is the vector with all their $3D$ positions stacked, and thus $n=15$. They must navigate in a 3D environment, avoiding collision with each other and the obstacles, and reach their designated setpoints.  The robots are represented by colored spheres, their setpoints by small cubes, and the obstacles are in black. The robots must navigate to their respective setpoints while avoiding collisions between themselves and the environment, as shown in Figure~ \ref{fig:envma} for the starting time. The obstacles (three tree-shaped obstacles and a wall-shaped obstacle) are modeled using primitives shapes: spheres, boxes, and cylinders. Overall, the obstacles are modeled using $3$ spheres, $3$ cylinders, and $4$ boxes, resulting in $10$ primitive obstacles.
\begin{figure}[htbp]
\includegraphics[width=8cm]{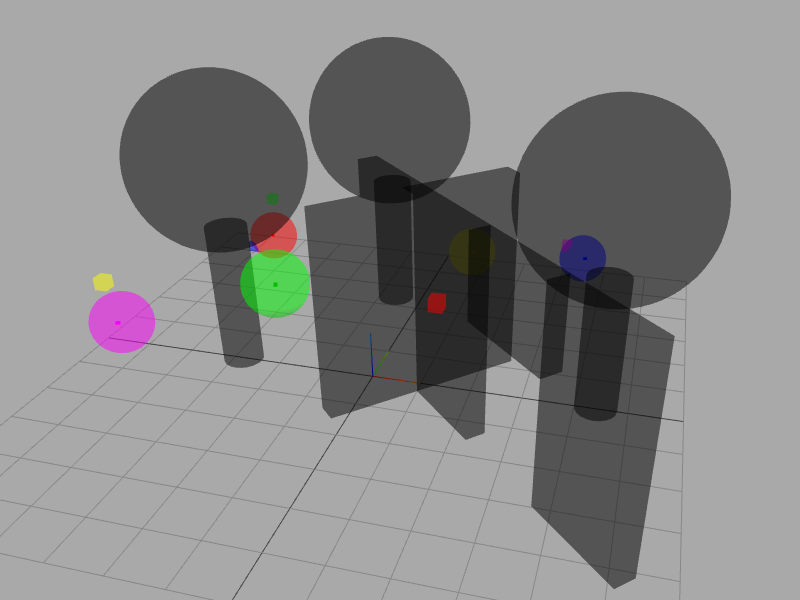}
\centering
\caption{Snapshot of the multi-robot simulation.}
\label{fig:envma}
\end{figure}

\begin{figure*}[ht]
    \centering
    \includegraphics[width=0.18\textwidth]{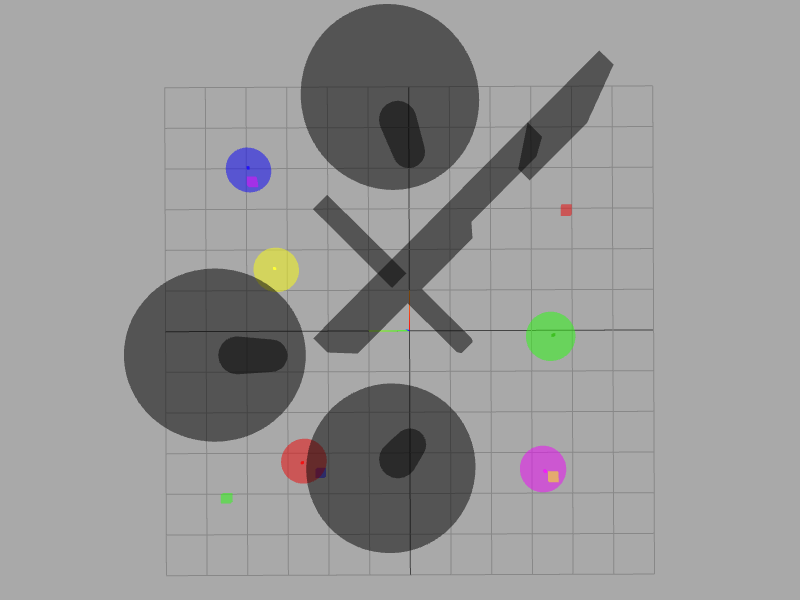}
    \includegraphics[width=0.18\textwidth]{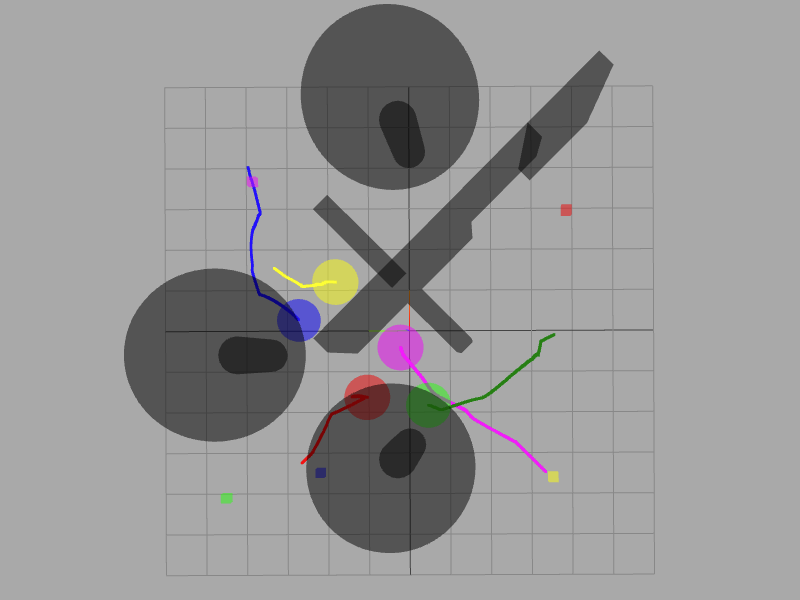}
    \includegraphics[width=0.18\textwidth]{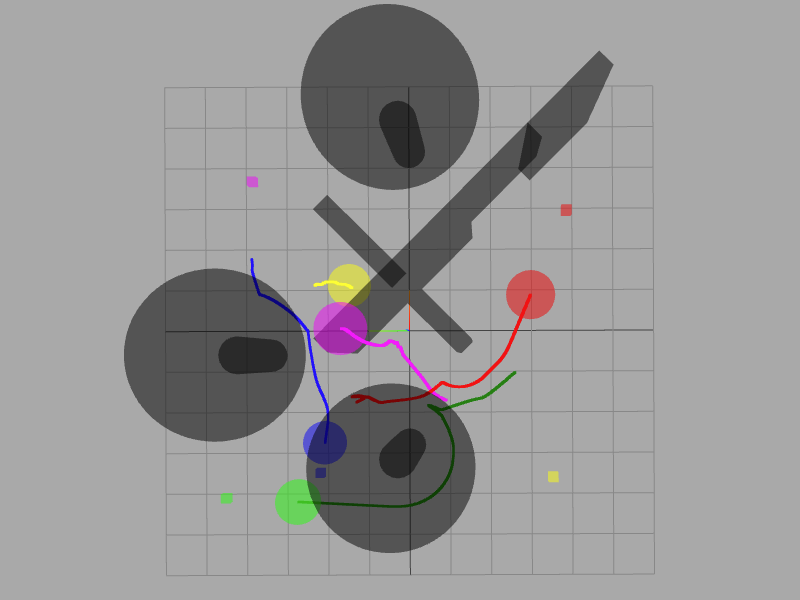}
    \includegraphics[width=0.18\textwidth]{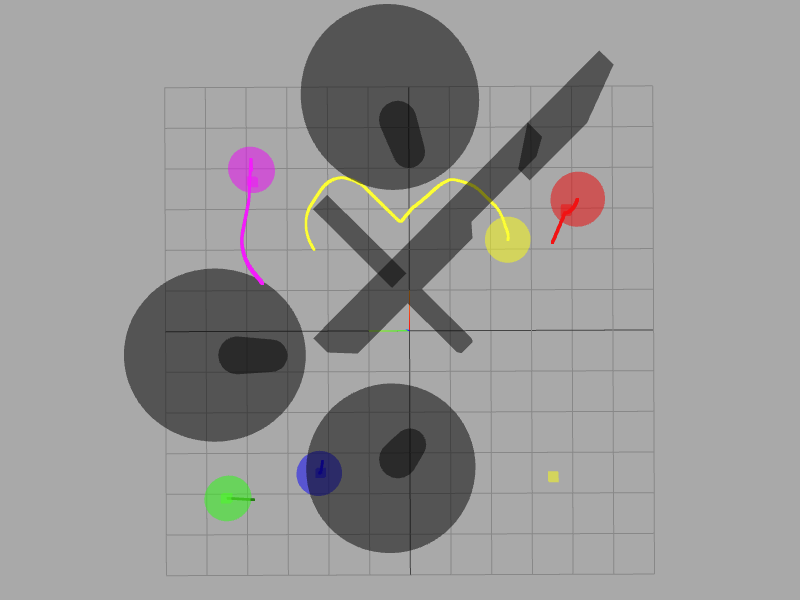}
    \includegraphics[width=0.18\textwidth]{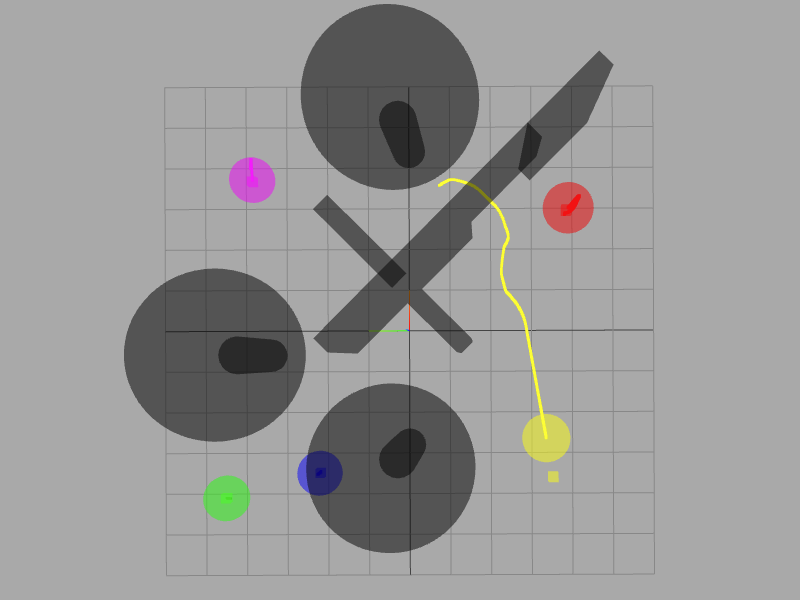}\\
    \vspace{3pt}
    \includegraphics[width=0.18\textwidth]{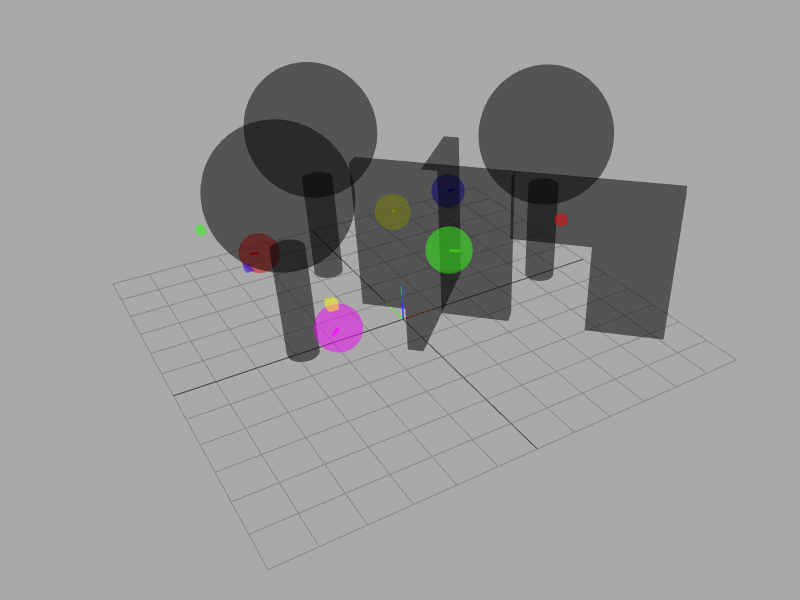}
    \includegraphics[width=0.18\textwidth]{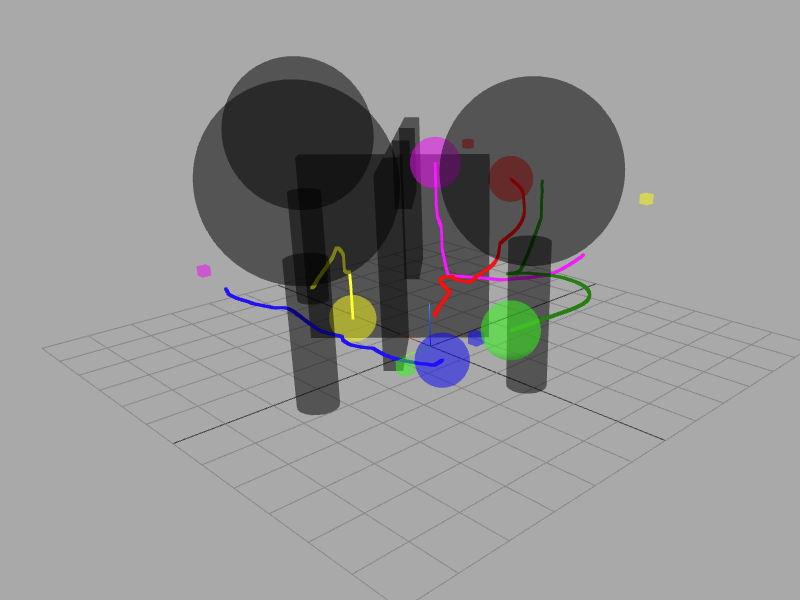}
    \includegraphics[width=0.18\textwidth]{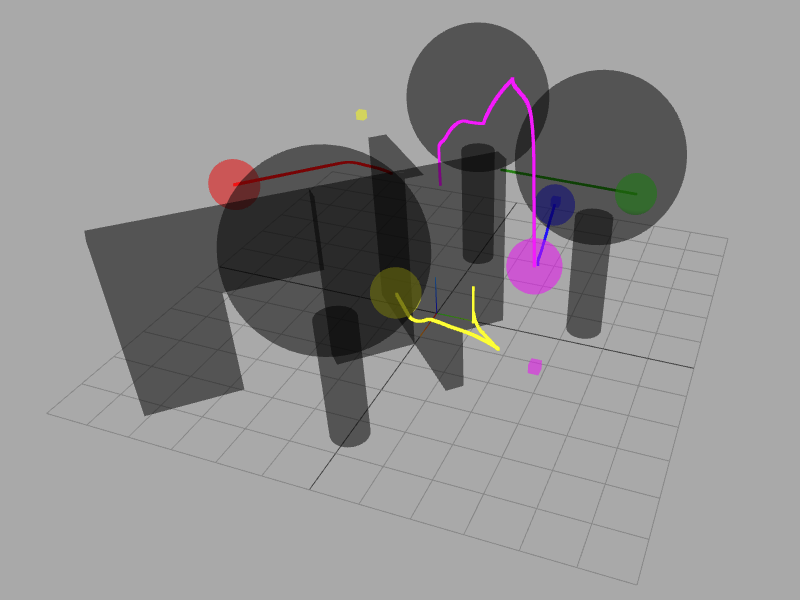}
    \includegraphics[width=0.18\textwidth]{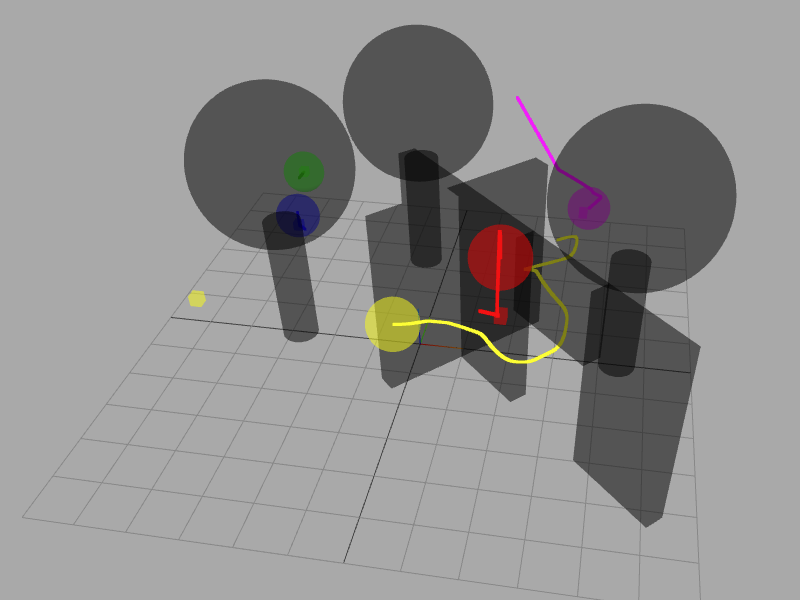}
    \includegraphics[width=0.18\textwidth]{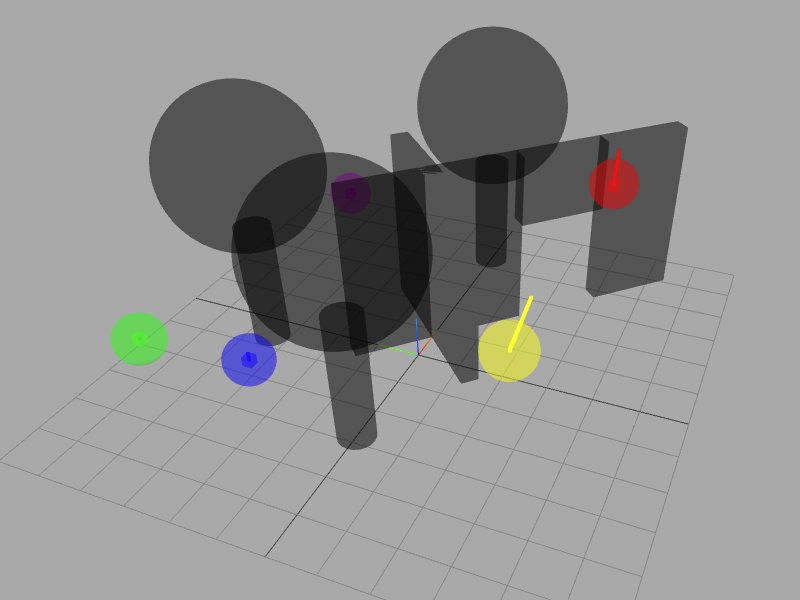}
    \caption{Snapshots for the multi-robot simulation. 
    }
    \label{fig:snapmultrobot}
\end{figure*}

Obstacles are modeled as spheres with radius $r=0.25m$, whereas for inter-agent collision, we model them as cylinders with this same radius and infinite height, with the cylinder axis normal to the ground. In this case, having one of them going below another is undesirable because the downwash from the top drone will push the bottom drone downwards. We then have $5 \times 10$ functions $F_i(q)$  that compute the distance between each robot and each primitive obstacle, $(5 \times 4) /2$ functions $F_i(q)$ that compute the distance between the robots, and finally $6 \times 5$ functions $F_i(q)$ that implement limits for the positions of the robots, since they must lie in a 6m$\times$6m$\times$3m box. Overall, we have $90$ functions $F_i$ that we condense into a single function $D$ using the method described in Subsection \ref{subs:thefunD} with $h = 0.01$ and $\delta=0.13$.

\begin{figure}[ht]
\includegraphics[width=8cm]{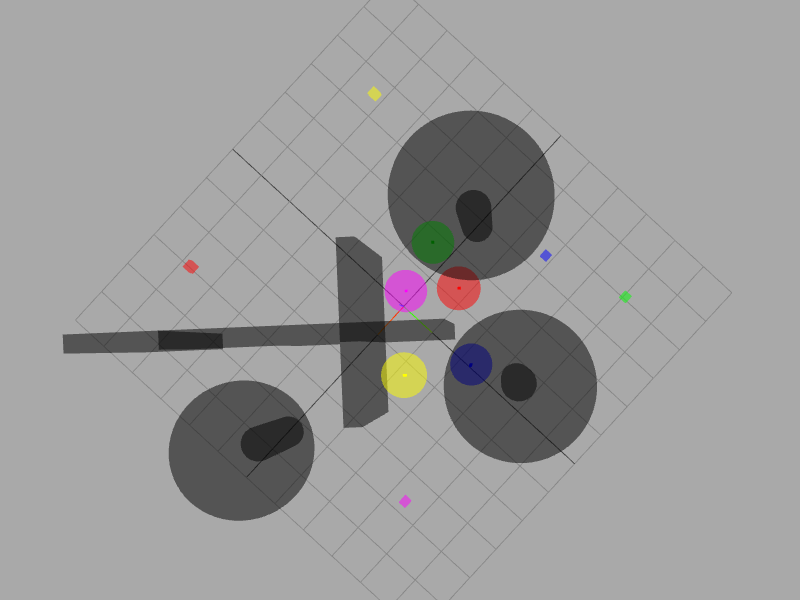}
\centering
\caption{Snapshot of the multi-robot simulation without the circulation inequality.}
\label{fig:multiagent_fail}
\end{figure}

\begin{figure}[ht]
\includegraphics[width=8cm,trim={3.8cm 10.5cm 4cm 11cm},clip]{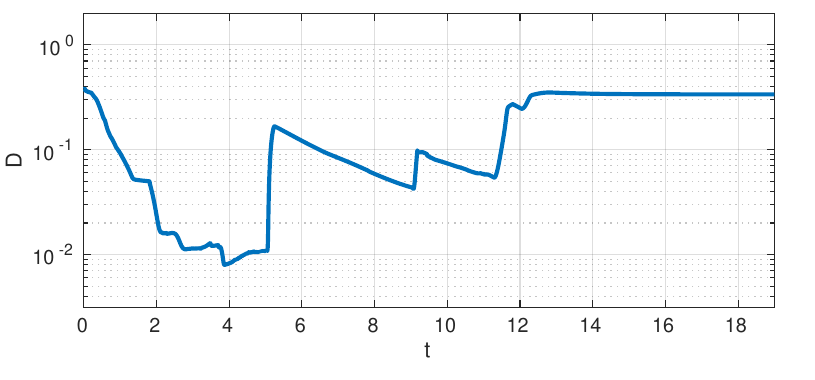}
\centering
\caption{Evolution of the distance function $D$.}
\label{fig:Dgraphmultiagent}
\end{figure}

Moreover, $\alpha(D) = -0.5 D$, $\beta(D) = 0.9 - 6D$, and $V(q) = \frac{1}{2} \sum_{i=1}^5 \|p_i-p_{g,i}\|^2$, in which $p_i \in \mathbb{R}^3$ is the position of the robots and $p_{g,i} \in \mathbb{R}^3$ their respective target positions. There are no constraints of the form $A \mu \geq b$. Since $n=15$ is odd-dimensional, we must implement the workaround discussed in Subsection \ref{subs:odddimensions}, where a skew-symmetric, not orthonormal/invertible matrix 
$\Omega \in \mathbb{R}^{15 \times 15}$ is employed, such that 
$\Omega = [N_2 \ \ {-}N_1 \ \ N_4 \ \ {-}N_3 \ \ldots \ N_{14} \ \ {-}N_{13} \ 0]^T$.

Figure \ref{fig:snapmultrobot} shows snapshots of the simulations. %The robots are represented by colored spheres, and their  setpoints are represented by a cube colored with their respective color. The obstacles are colored in black. 
In the first row,  5 snapshots are shown, in sequential timeframes, using a top view of the path of each robots. In the second row, we show the same timeframes as the first row, but with a different view. We can see circulation-like behaviors that are essential to accomplish the task and that would not be possible only with the CBF constraint. For example, the yellow robot in the fourth snapshot circulates the cross-shaped obstacle in order to reach its setpoint.

Indeed, Figure~\ref{fig:multiagent_fail} shows the result of applying the controller without the circulation inequality, in which we can see that the robots get stuck into a spurious equilibrium point. Figure~\ref{fig:Dgraphmultiagent} shows the evolution of the positive function $D$. %An interactive 3d visualization of this simulation can be seen \href{https://raw.githack.com/viniciusmgn/uaibot_content/master/contents/Simulation/cbfcirc_cdc.html}{here}.

\section{Conclusion}
\label{sec:conclusion}
An optimization framework was proposed for controlling a system with integrator dynamics to a goal in the configuration space while considering obstacle avoidance in addition to other constraints. The novelty of the framework is the introduction of a circulation inequality that forces the system to circulate an obstacle - in the configuration space - when sufficiently near it. The motivation for the approach is that often CBF-QP formulations are prone to reach spurious equilibrium points instead of the desired goal. The introduced inequality allows us to extend the range of situations in which the approach is successful with a negligible increase in the complexity of the formulation. We showed several formal results for the controller including continuity, feasibility, and classification of the equilibrium points, and also demonstrated the efficacy of the approach using simulation studies.

\bibliographystyle{ieeetr}  
\bibliography{bibl}
\end{document}